\documentclass{comsoc2010}

\usepackage[amsthm]{pamas}

\usepackage{pgflibraryshapes}  
\usetikzlibrary{positioning}  
\usetikzlibrary{patterns,arrows,shapes}

\usepackage{algorithm,algorithmic} 

\newcommand{\exist}{\text{{\sc Existence}}\xspace}
\newcommand{\construct}{\text{{\sc Search}}\xspace}
\newcommand{\optimality}{\text{{\sc Optimality}}\xspace}

\newcommand{\pref}{\mathcal{P}\xspace}

\newcommand{\ASHG}{ASHG\xspace}
\newcommand{\ASHGS}{ASHGs\xspace}

\newcommand\eat[1]{}

\title{Optimal Partitions in Additively Separable Hedonic Games\footnote{A preliminary version of this work was invited for presentation in the session `Cooperative Games and Combinatorial Optimization' at the 24th European Conference on Operational Research (EURO 2010) in Lisbon. This material is based on work supported by the Deutsche Forschungsgemeinschaft under grants BR-2312/6-1 (within the European Science Foundation's EUROCORES program LogICCC) and BR~2312/7-1.}}
\author{Haris Aziz, Felix Brandt, and Hans Georg Seedig}

\begin{document}

\begin{abstract}
We conduct a computational analysis of fair and optimal partitions in additively separable hedonic games.
We show that, for strict preferences, a Pareto optimal partition can be found in polynomial time while verifying whether a given partition is Pareto optimal is coNP-complete, even when preferences are symmetric and strict. Moreover, computing a partition with maximum egalitarian or utilitarian social welfare or one which is both Pareto optimal and individually rational is NP-hard. 
We also prove that checking whether there exists a partition which is both Pareto optimal and envy-free is $\Sigma_{2}^{p}$-complete. Even though an envy-free partition and a Nash stable partition are both guaranteed to exist for symmetric preferences, checking whether there exists a partition which is both envy-free and Nash stable is NP-complete.
\end{abstract}

\section{Introduction}

Ever since the publication of \citeauthor{vNM47a}'s \emph{Theory of Games and Economic Behavior} in 1944, coalitions have played a central role within game theory. The crucial questions in coalitional game theory are
which coalitions can be expected to form and how the members of coalitions should divide the proceeds of their cooperation. Traditionally the focus has been on the latter issue, which led to the formulation and analysis of concepts such as Gillie's core, the Shapley value, or the bargaining set.
Which coalitions are likely to form is commonly assumed to be settled exogenously, either by explicitly specifying the coalition structure, a partition of the players in disjoint coalitions, or, implicitly, by assuming that larger coalitions can invariably guarantee better outcomes to its members than smaller ones and that, as a consequence, the grand coalition of all players will eventually form.

The two questions, however, are clearly interdependent: the individual players' payoffs depend on the coalitions that form just as much as the formation of coalitions depends on how the payoffs are distributed. 

\emph{Coalition formation games}, as introduced by \citet{DrGr80a}, provide a simple but versatile formal model that allows one to focus on coalition formation as such. In many situations it is natural to assume that a player's appreciation of a coalition structure only depends on the coalition he is a member of and not on how the remaining players are grouped.
Initiated by  \citet{BKS01a} and \citet{BoJa02a}, much of the work on coalition formation now concentrates on these so-called \emph{hedonic games}. 

The main focus in hedonic games has been on notions of \emph{stability} for coalition structures such as Nash stability, individual stability, contractual individual stability, or core stability and characterizing conditions under which they are guaranteed to be non-empty \citep[see, \eg][]{BoJa02a}. 
The most prominent examples of hedonic games are two-sided matching games in which only coalitions of size two are admissible~\citep{RoSo90a}.

General coalition formation games have also received attention from the artificial intelligence community, where the focus has generally been on computing partitions that give rise to the greatest social welfare \citep[see, \eg][]{SLA+99a}. 
The computational complexity of hedonic games has been investigated with a focus on the complexity of computing stable partitions for different models of hedonic games~\citep{Ball04a,DBHS06a,Cech08a}. We refer to~\citet{Hajd06a} for a critical overview.

Among hedonic games, \emph{additively separable hedonic games (ASHGs)} are a particularly natural and succinct representation in which each player has a value for every other player and the value of a coalition to a particular player is computed by simply adding his values of the players in his coalition.

Additive separability satisfies a number of desirable axiomatic properties~\citep{BBP04a}. 
\ASHGS are the non-transferable utility generalization of \emph{graph games} studied by \citet{DePa94a}. 
\citet{SuDi10a} showed that for \ASHGS, checking whether a core stable, strict-core stable, Nash stable, or individually stable partition exists is NP-hard. \citet{DBHS06a} obtained positive algorithmic results for subclasses of additively separable hedonic games in which each player divides other players into friends and enemies. \citet{BrLa09a} examined the tradeoff between stability and social welfare in \ASHGS.


\paragraph{Contribution}
In this paper, we analyze concepts from fair division in the context of coalition formation games. 
We present the first systematic examination of the complexity of computing and verifying optimal partitions of hedonic games, specifically \ASHGS. We examine various standard criteria from the social sciences: \emph{Pareto optimality}, \emph{utilitarian social welfare}, \emph{egalitarian social welfare}, and \emph{envy-freeness}~\citep[see, \eg][]{Moul88a}. 

In Section~\ref{sec:sw}, we show that computing a partition with maximum egalitarian social welfare is NP-hard. Similarly, computing a partition with maximum utilitarian social welfare is NP-hard in the strong sense even when preferences are symmetric and strict. 

In Section~\ref{sec:PO}, the complexity of Pareto optimality is studied. 
We prove that checking whether a given partition is Pareto optimal is coNP-complete in the strong sense, even when preferences are strict and symmetric.  By contrast, we present a polynomial-time algorithm for computing a Pareto optimal partition when preferences are strict.\footnote{Thus, we identify a natural problem in coalitional game theory where verifying a possible solution is presumably harder than actually finding one.} Interestingly, computing an individually rational \emph{and} Pareto optimal partition is NP-hard in general.


In Section~\ref{sec:envy}, we consider complexity questions regarding envy-free partitions. Checking whether there exists a partition which is both Pareto optimal and envy-free is shown to be $\Sigma_{2}^{p}$-complete. We present an example which exemplifies the tradeoff between satisfying stability (such as Nash stability) and envy-freeness and use the example to prove that checking whether there exists a partition which is both envy-free and Nash stable is NP-complete even when preferences are symmetric.

Our computational hardness results imply computational hardness of equivalent problems for \emph{hedonic coalition nets} \citep{Elwo09a}.

\section{Preliminaries}

In this section, we provide the terminology and notation required for our results.

\subsection{Hedonic games}

A \emph{hedonic coalition formation game} is a pair $(N,\pref)$ where $N$ is a set of players and $\pref$ is a \emph{preference profile} which specifies for each player $i\in N$ the preference relation $ \succsim_i$, a reflexive, complete and transitive binary relation on set $\mathcal{N}_i=\{S\subseteq N \mid i\in S\}$.

$S\succ_iT$ denotes that $i$ strictly prefers $S$ over $T$ and $S\sim_iT$ that $i$ is indifferent between coalitions $S$ and $T$. A \emph{partition} $\pi$ is a partition of players $N$ into disjoint coalitions. By $\pi(i)$, we denote the coalition in $\pi$ which includes player $i$.

A game $(N,\pref)$ is \emph{separable} if for any player $i\in N$ and any coalition $S\in \mathcal{N}_i$ and for any player $j$ not in $S$ we have the following: $S\cup\{j\}\succ_i S$ if and only if $\{i,j\}\succ_i \{i\}$; $S\cup\{j\}\prec_i S$ if and only if $\{i,j\}\prec_i \{i\}$; and  $S\cup\{j\}\sim_i S$ if and only if $\{i,j\}\sim_i \{i\}$.



In an \emph{additively separable hedonic game} $(N,\pref)$, each player $i\in N$ has value $v_i(j)$ for player $j$ being in the same coalition as $i$ and if $i$ is in coalition $S\in \mathcal{N}_i$, then $i$ gets utility $\sum_{j\in S\setminus \{i\}}v_i(j)$. For coalitions $S,T\in\mathcal{N}_i$, $S \succsim_i T$ if and only if $\sum_{j\in S\setminus \{i\}}v_i(j) \geq \sum_{j\in T\setminus \{i\}}v_i(j)$.


A preference profile is \emph{symmetric} if $v_i(j)=v_j(i)$ for any two players $i,j\in N$ and is \emph{strict} if $v_i(j)\neq 0$ for all $i,j\in N$ such that $i\neq j$. We consider \ASHGS (additively separable hedonic games) in this paper. Unless mentioned otherwise, all our results are for \ASHGS. 

\eat{
\begin{definition}
	A profile $\mathcal{P}$ of additively separable preferences $(v_1,\ldots v_n)$ satisfies 
	\begin{itemize}
		\item \emph{symmetry} if $v_i(j)=v_j(i)$ for all $i,j\in N$. 
		\item \emph{strictness} if $v_i(j)>0$ or $v_i(j)<0$ for all $i,j\in N$.
	\end{itemize}
\end{definition}

For any player $i$, let $F(i)=\{j \mid v_i(j)> 0\}$ be the set of players which $i$ strictly likes. Similarly, let $E(i)=\{j \mid v_i(j)> 0\}$ be the set of players which $i$ strictly dislikes
}

\subsection{Fair and optimal partitions}

In this section, we formulate concepts from the social sciences, especially the literature on fair division, for the context of hedonic games. 
A partition $\pi$ satisfies \emph{individual rationality} if each player does as well as by being alone, i.e., for all $i\in N$, $\pi(i) \succsim_i  \{i\}$.
For a utility-based hedonic game $(N,\pref)$ and partition $\pi$, we will denote the utility of player $i\in N$ by $u_{\pi}(i)$. 
%
The different notions of fair or optimal partitions are defined as follows.\footnote{All welfare notions considered in this paper (utilitarian, elitist, and egalitarian) are based on the interpersonal comparison of utilities.  Whether this assumption can reasonably be made is debatable.}

\begin{enumerate}
	\item The \emph{utilitarian social welfare} of a partition is defined as the sum 
	of individual utilities of the players: $u_{ut}(\pi)=\sum_{i\in N}u_{\pi}(i)$. A \emph{maximum utilitarian partition} maximizes the utilitarian social welfare.
	\item The \emph{elitist social welfare} is given by the utility of the player that is best off: $u_{el}(\pi)=\max \{u_{\pi}(i)\mid i\in N\}$. A \emph{maximum elitist partition} maximizes the utilitarian social welfare.
		\item The \emph{egalitarian social welfare} is given by the utility of the agent that is worst off: $u_{eg}(\pi)=\min \{u_{\pi}(i) \mid i\in N\}$. A \emph{maximum egalitarian partition} maximizes the egalitarian social welfare.
		\item A partition $\pi$ of $N$ is \emph{Pareto optimal} if there exists no partition $\pi'$ of $N$ which \emph{Pareto dominates} $\pi$, that is for all $i\in N$,  $\pi'(i) \succsim_i  \pi(i)$ and there exists at least one player $j\in N$ such that $j\in N$,  $\pi'(j) \succ_j  \pi(j)$. 
	\item \emph{Envy-freeness} is a notion of fairness. In an \emph{envy-free} partition, no player has an incentive to replace another player. 
\end{enumerate}

For the sake of brevity, we will call all the notions described above  ``optimality criteria'' although envy-freeness is rather concerned with fairness than optimality. We consider the following computational problems with respect to the optimality criteria defined above. \\


\eat{
Envy-freeness appears similar to Nash stability in which no player has an incentive to move to another coalition. However, one can produce simple examples to show that envy-freeness does not imply Nash stability and Nash stability does not imply envy-freeness.\begin{example}
	A partition which satisfies envy-freeness may not be Nash stable. Take the game $(N,\pref)$ where $N=\{1,2\}$ and where $\pref$ is specified by $v_1(2)=v_2(1)=1$. Then the partition $\pi=\{\{1\},\{2\}\}$ satisfies envy-freeness but it is not Nash stable. 
Similarly, a Nash stable partition may not satisfy envy-freeness. Take the game $(N,\pref)$ where $N=\{1,2,3\}$ where $\pref$ is specified by: $v_1(2)=1$, $v_1(3)=-1$, $v_2(3)=v_3(2)=2$ and $v_2(1)=v_3(1)=0$. Consider the partition $\pi=\{\{1\}, \{2,3\}\}$ which is Nash stable. However, $\pi$ does not satisfy envy-freeness because player $1$ is envious of player $3$ and would prefer to replace it to be with player $2$.	
\end{example}
}

\noindent
\optimality: Given $(N,\pref)$ and a partition $\pi$ of $N$, is $\pi$ optimal?\\
\exist: Does an optimal partition for a given $(N,\pref)$ exist?\\
\construct: If an optimal partition for a given $(N,\pref)$ exists, find one.\\

\exist is trivially true for all criteria of optimality concepts. By the definitions, it follows that there exist partitions which satisfy maximum utilitarian social welfare, elitist social welfare, and egalitarian social welfare respectively. 

\eat{
\begin{remark}
The partition consisting of the grand coalition and the partition of singletons satisfy the envy-freeness property. Therefore, for any representation of hedonic games, a partition which satisfies envy-freeness can computed efficiently. 
\label{remark:envy-easy}
\end{remark}

\begin{fact}
	For any hedonic game, the problem \optimality for any notion of efficiency is in $\Sigma_{2}^P$.
\end{fact}
\begin{proof}
	A partition which is optimal according to any notion of efficiency has to achieve social welfare as good as every other partition. 
\end{proof}
}

\section{Complexity of maximizing social welfare}\label{sec:sw}

In this section, we examine the complexity of maximizing social welfare in \ASHGS. 
We first observe that computing a maximum utilitarian partition for strict and symmetric preferences is NP-hard because it is equivalent to the NP-hard problem of maximizing agreements in the context of correlation clustering~\citep{BBC04a}. 
 
\begin{theorem}
Computing a maximum utilitarian partition is NP-hard in the strong sense even with symmetric and strict preferences.
\label{prop:util-hard}
\end{theorem}
Computing a maximum elitist partition is much easier. For any player $i$, let $F(i)=\{j \mid v_i(j)> 0\}$ be the set of players which $i$ strictly likes and $f(i)=\sum_{j\in F(i)}v_i(j)$. Both $F(i)$ and $f(i)$ can be computed in linear time. Let $k\in N$ be the player such that $f(k)\geq f(i)$ for all $i\in N$. Then $\pi=\{\{\{k\}\cup F(k)\}, N\setminus \{\{k\}\cup F(k)\}\}$ is a partition which maximizes the elitist social welfare. As a corollary, we can verify whether a partition $\pi$ has maximum elitist social welfare by computing a partition $\pi^*$ with maximum elitist social welfare and comparing $u_{el}(\pi)$ with $u_{el}(\pi^*)$. Just like maximizing the utilitarian social welfare, maximizing the egalitarian social welfare is hard.

\eat{
\begin{proposition}
There exists a polynomial-time algorithm to compute a maximum elitist partition.
\label{prop:elite-easy}
\end{proposition}
\begin{proof}
Recall that for any player $i$, $F(i)=\{j \mid v_i(j)> 0\}$. Let $f(i)=\sum_{j\in F(i)}v_i(j)$. Both $F(i)$ and $f(i)$ can be computed in linear time. Let $k\in N$ be the player such that $f(k)\geq f(i)$ for all $i\in N$. Then 
$\pi=\{\{\{k\}\cup F(k)\}, N\setminus \{\{k\}\cup F(k)\}\}$ is a partition which maximizes the elitist social welfare. 
\end{proof}

As a corollary, we can verify whether a partition $\pi$ has maximum elitist social welfare by computing a partition $\pi^*$ with maximum elitist social welfare and comparing $u_{el}(\pi)$ with $u_{el}(\pi^*)$. Just like maximizing the utilitarian social welfare, maximizing the egalitarian social welfare is hard:
}

\begin{theorem}\label{prop:egal-hard}
Computing a maximum egalitarian partition is NP-hard in the strong sense. 
\end{theorem}
\begin{proof}
We provide a polynomial-time reduction from the NP-hard problem {\sc MaxMinMachineCompletionTime}~\citep{Woeg97a} in which an instance consists of a set of $m$ identical machines $M=\{M_1,\ldots, M_m\}$, a set of $n$ independent jobs $J= \{J_1,\ldots,J_n\}$  where job $J_i$ has processing time $p_i$. The problem is to allot jobs to the machines such that the minimum processing time (without machine idle times) of all machines is maximized.
\eat{
	We provide a polynomial-time reduction from the following NP-hard problem~\citep{Woeg97a}: \\
	
\noindent
{\sc MaxMinMachineCompletionTime} \\
\noindent
INSTANCE: A set of $m$ identical machines $M=\{M_1,\ldots, M_m\}$, a set of $n$ independent jobs $J= \{J_1,\ldots,J_n\}$  where job $J_i$ has processing time $p_i$.\\
\noindent
OUTPUT: Allot jobs to the machines such that the minimum processing time (without machine idle times) of all machines is maximized. \\
	}
Let $I$ be an instance of {\sc MaxMinMachineCompletionTime} and let $P=\sum_{i=1}^n p_i$. From $I$ we construct an instance $I'$ of {\sc EgalSearch}. The \ASHG for instance $I'$ consists of $N=\{i \mid M_i\in M\}\cup \{s_i \mid  J_i\in J\}$ and the preferences of the players are as follows: for all $i=1,\ldots m$ and all $j=1,\ldots, n$ let $v_i(s_j)=p_j$ and $v_{s_j}(i)=P$. Also, for $1\leq i,i' \leq m, i\neq i'$ let $v_i(i')=-(P+1)$ and for $1\leq j,j' \leq n, j\neq j'$ let $v_{s_j}(v_{s_{j'}})=0$. Each player $i$ corresponds to machine $M_i$ and each player $s_j$ corresponds to job $J_j$. 
	
	Let $\pi$ be the partition which maximizes $u_{eg}(\pi)$. We show that players $1,\ldots, m$ are in separate coalitions and each player $s_j$ is in $\pi(i)$ for some $1\leq i\leq m$. We can do so by proving two claims. 
	The first claim is that for $i,j\in \{1,\ldots m\}$ such that $i\neq j$, we have that $i\notin \pi(j)$. The second claim is that each player $s_j$ is in a coalition with a player $i$. The proofs of the claims are omitted due to space limitations.
	
	\eat{
Let $\pi$ be the partition which maximizes $u_{eg}(\pi)$. We show that players $1,\ldots, m$ are in separate coalitions and each player $s_j$ is in $\pi(i)$ for some $1\leq i\leq m$. We do so by proving two claims. 
The first claim is that for $i,j\in \{1,\ldots m\}$ such that $i\neq j$, we have that $i\notin \pi(j)$. 
Assume there exist exactly two players $i$ and $j$ for which this is not the case. Then we know that $u_{\pi}(i)=-(P+1) +\sum_{s_j\in \pi(i)}p_j$. Since $\sum_{s_j\in \pi(i)}p_j\leq P$, we know that $u_{\pi}(i)=u_{\pi}(j)<0$, $u_{\pi}(a)\geq 0$ for all $a\in N\setminus \{i,j\}$ and thus $u_{eg}(\pi)<0$. However, if $i$ deviates and forms 
a singleton coalition in new partition $\pi'$, then $u_{\pi'}(i)=0$ and $u_{\pi'}(j)\geq 0$ and the utility of other players has not decreased. Therefore, $u_{eg}(\pi')\geq 0$ which is a contradiction.

The second claim is that each player $s_j$ is in a coalition with a player $i$. Assume this was not the case so that there exists at least one such player $s_j$. Since we already know that all $i$s are in separate coalitions, then $u_{\pi}(a)>0$ for all $a\in N\setminus\{s_j\}$ and $u_{eg}(\pi)=u_{\pi}(s_j)=0$. Then $s_j$ can deviate and join $\pi(i)$ for any $1\leq i\leq m$ to form a new partition $\pi'$. By that, the utility of no player decreases and $u_{\pi'}(s_j)>0$. If this is done for all such $s_j$, we have $u_{eg}(\pi')> 0$ for the new partition $\pi'$ which is a contradiction. 
}

A job allocation $\mathrm{Alloc}(\pi)$ corresponds to a partition $\pi$ where $s_j$ is in $\pi(i)$ if job $J_j$ is assigned to $M_i$ for all $j$ and $i$. Note that the utility $u_{\pi}(i)=\sum_{s_j\in \pi(i)}v_i(s_j)=\sum_{s_j\in \pi(i)}p_j$ of a player corresponds to the total completion time of all jobs assigned to $M_i$ according to $\mathrm{Alloc}(\pi)$. Let $\pi^*$ be a maximum egalitarian partition. Assume that there is another partition $\pi'$ and $\mathrm{Alloc}(\pi')$ induces a strictly greater minimum completion time. We know that $u_{\pi^*}(s_j)=u_{\pi''}(s_j)=P$ for all $1\leq j \leq n$ and $u_{\pi^*}(i)\leq P$ for all $1\leq i \leq m$. But then from the assumption we have $u_{eg}(\pi')>u_{eg}(\pi^*)$ which is a contradiction.\end{proof}


\section{Complexity of Pareto optimality}\label{sec:PO}

We now consider the complexity of computing a Pareto optimal partition.
The complexity of Pareto optimality has already been considered in several settings such as house allocation~\citep{ACMM05a}. \citet{BoLa08a} examined the complexity of Pareto optimal allocations in resource allocation problems. We show that checking whether a partition is Pareto optimal is hard even under severely restricted settings.

\eat{
\begin{example}
	Consider a hedonic game $(N,\pref)$ where $N=\{1,2,3\}$ and the $\pref$ is defined as follows:$v_1(2)=v_1(3)=2$, $v_2(1)=v_3(1)=-1$ and $v_2(3)=v_3(2)=10$. 
Consider the partition $\pi$ consisting of the grand coalition. Then $u_{\pi}(1)=4$, $u_{\pi}(2)=9$ and $u_{\pi}(9)=4$. It is easy to see that $\pi$ is Pareto optimal. The players $2$ and $3$ improve their welfare to $10$ each if they break away and form the coalition $\{2,3\}$. However, player $1$ is worse off and gets utility $0$.
\end{example}
}

\begin{theorem}
The problem of checking whether a partition is Pareto optimal is coNP-complete in the strong sense, even when preferences are symmetric and strict. 
\label{prop:check-PO-symm-hard}
\end{theorem}
\begin{proof}

	The reduction is from the NP-complete problem E3C (EXACT-3-COVER) to deciding whether a given partition is Pareto dominated by another partition or not. Recall that in E3C, an instance is a pair $(R,S)$, where $R=\{1,\ldots, r\}$ is a set and $S$ is a collection of subsets of 
	$R$ such that $|R|=3m$ for some positive integer $m$ and $|s| = 3$ for each 
	$s\in S$. The question is whether there is a sub-collection $S'\subseteq S$ which is a partition of $R$.
	
	\eat{
	We recall the E3C problem.\\
	
	\noindent
	{\sc E3C (EXACT-3-COVER)}: \\
	\noindent
	INSTANCE: A pair $(R,S)$, where $R=\{1,\ldots, r\}$ is a set and $S$ is a collection of subsets of 
	$R$ such that $|R|=3m$ for some positive integer $m$ and $|s| = 3$ for each 
	$s\in S$. \\
	\noindent
	QUESTION: Is there a sub-collection $S'\subseteq S$ which is a partition of $R$? \\
}
	It is known that E3C remains NP-complete even if each $r\in R$ occurs in 
	at most three members of $S$~\citep{GaJo79a}. 
	Let $(R,S)$ be an instance of E3C. 
	$(R,S)$ can be reduced to an instance $((N,\pref),\pi)$, where $(N,\pref)$ is an \ASHG defined in the following way. Let $N=\{w^s, x^s, y^s \mid  s \in S\} \cup \{z^r \mid r\in R \}$. The players preferences are symmetric and strict and are defined as follows:
$v_{w^s}(x^s)= v_{x^s}(y^s)=3$ for all $s\in S$; $v_{y^s}(w^s)=v_{y^s}(w^{s'})=-1$ for all $s,s'\in S$; $v_{y^s}(z^r)=1$ if $r\in s$  and $v_{y^s}(z^r)=-7$ if $r\notin s$; $v_{z^r}(z^{r'})=1/(|R|-1)$ for any $r,r'\in R$; and $v_{a}(b)=-7$ for any $a,b\in N$ and $a\neq b$ for which $v_{a}(b)$ is not already defined.

\eat{
	\begin{itemize}
	\item $v_{w^s}(x^s)= v_{x^s}(y^s)=3$ for all $s\in S$
	\item $v_{y^s}(w^s)=v_{y^s}(w^{s'})=-1$ for all $s,s'\in S$ 
	\item $v_{y^s}(z^r)=1$ if $r\in s$  and $v_{y^s}(z^r)=-1$ if $r\notin s$ and 
	\item $v_{z^r}(z^{r'})=1/(|R|-1)$ for any $r,r'\in R$
    \item $v_{a}(b)=-7$ for any $a,b\in N$ and $a\neq b$ for which $v_{a}(b)$ is not already defined,
    \end{itemize}
}

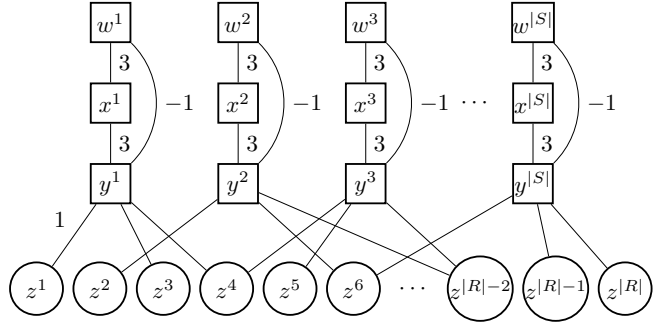
\begin{figure}[tb]
	\centering
	\scalebox{0.87}{
	\begin{tikzpicture}[auto, player/.style={circle,draw=black!100,fill=black!00,thick,minimum size=8mm,inner sep=0pt, node distance=0.13},
	3set/.style={rectangle,draw=black!100,fill=black!00,thick,minimum size=6mm,inner sep=0pt,node distance=1.3},
	control/.style={rectangle,draw=black!100,fill=black!00,thick,minimum size=6mm,inner sep=0pt, node distance=0.6},scale=0.8]

	\node[player] (z1)  at (0,0) {$z^1$};
	\node[player] (z2) [right=of z1] {$z^2$};
	\node[player] (z3) [right=of z2] {$z^3$};
	\node[player] (z4) [right=of z3] {$z^4$};
	\node[player] (z5) [right=of z4] {$z^5$};
	\node[player] (z6) [right=of z5] {$z^6$};
	\node[node distance = 0.15] (player-dots) [right=of z6] {$\cdots$};
	\node[player] (zr-2) [right=of player-dots] {$z^{|R|-2}$};
	\node[player] (zr-1) [right=of zr-2] {$z^{|R|-1}$};
	\node[player] (zr) [right=of zr-1] {$z^{|R|}$};

	\node[3set] (y1) at (1.4,2) {$y^1$};
	\node[3set, right=of y1] (y2) {$y^2$};
	\node[3set, right=of y2] (y3) {$y^3$};
	\node[3set, node distance = 1.9, right=of y3] (ym) {$y^{|S|}$};

	\node[control] (x1) [above=of y1] {$x^1$};
	\node[control] (x2) [above=of y2] {$x^2$};
	\node[control] (x3) [above=of y3] {$x^3$};
	\node[control] (xm) [above=of ym] {$x^{|S|}$};
	\node[control] (w1) [above=of x1] {$w^1$};
	\node[control] (w2) [above=of x2] {$w^2$};
	\node[control] (w3) [above=of x3] {$w^3$};
	\node[control] (wm) [above=of xm] {$w^{|S|}$};

	\node[node distance = 1.0, right=of x3] (x-dots) {$\cdots$};

	\draw (y1)    to  node [swap ]{$1$} (z1);
	\draw[-] (y1) to (z3);
	\draw[-] (y1) to (z4);
	\draw[-] (y2) to (z2);
	\draw[-] (y2) to (z6);
	\draw[-] (y2) to (zr-2);
	\draw[-] (y3) to (z4);
	\draw[-] (y3) to (z5);
	\draw[-] (y3) to (zr-2);
	\draw[-] (ym) to (z6);
	\draw[-] (ym) to (zr-1);
	\draw[-] (ym) to (zr);

	\draw[-] (w1) to node {$3$} (x1);
	\draw[-] (w2) to node {$3$} (x2);
	\draw[-] (w3) to node {$3$} (x3);
	\draw[-] (wm) to node {$3$} (xm);

	\draw[-] (x1) to node {$3$} (y1);
	\draw[-] (x2) to node {$3$} (y2);
	\draw[-] (x3) to node {$3$} (y3);
	\draw[-] (xm) to node {$3$} (ym);

	\draw[-] (w1) to [bend left=45] node {$-1$} (y1) ;
	\draw[-] (w2) to [bend left=45] node {$-1$} (y2) ;
	\draw[-] (w3) to [bend left=45] node {$-1$} (y3) ;
	\draw[-] (wm) to [bend left=45] node {$-1$} (ym) ;

	\end{tikzpicture}
}
	\scriptsize
	\caption{A graph representation of an ASHG derived from an instance of E3C. The (symmetric) utilities are given as edge weights. Some edges and labels are omitted: All edges between any $y^s$ and $z^r$ have weight $1$ if $r\in s$. All $z^{r'},z^{r''}$ with $r'\neq r''$ are connected with weight $\frac{1}{|R|-1}$. All other edges missing in the complete undirected graph have weight $-7$. }
\end{figure}

The partition $\pi$ in the instance $((N,\pref),\pi)$ is 
$\{\{x^s, y^s\}, \{w^s\} \mid s \in S\}\}\cup \{\{z^r\mid r\in R\}\}.$
We see that the utilities of the players are as follows: $u_{\pi}(w^s)= 0$ for all $s\in S$; $u_{\pi}(x^s)= u_{\pi}(y^s)=3$ for all $s\in S$; and $u_{\pi}(z^r)= 1$ for all $r\in R$.

Assume that there exists $S'\subseteq S$ such that $S'$ is a partition of $R$. Then we prove that $\pi$ is not Pareto optimal and there exists another partition $\pi'$ of $N$ which Pareto dominates $\pi$.
We form another partition $\pi'=\{\{x^s, w^s\}\mid s \in S'\}\cup \{\{y^s,z_i, z_j, z_k\}\mid s\in S' \wedge i, j, k\in s\} \cup \{\{x^s, y^s\}, \{w^s\}\mid s \in (S\setminus S')\}\}.$

In that case, $u_{\pi'}(w^s)= 3$ for all $s\in S'$; $u_{\pi'}(w^s)= 0$ for all $s\in S\setminus S'$; $u_{\pi}(x^s)= u_{\pi}(y^s)=3$ for all $s\in S$; and $u_{\pi}(z^r)= 1+2/(|R|-1)$ for all $r\in R$. Whereas the utilities of no player in $\pi'$ decreases, the utility of some players in $\pi'$ is more than in $\pi$. Since $\pi'$ Pareto dominates $\pi$, $\pi$ is not Pareto optimal.

We now show that if there exists no $S'\subseteq S$ such that $S'$ is a partition of $R$, then $\pi$ is Pareto optimal. We note that $-7$ is a sufficiently large negative valuation to ensure that if $v_a(b)=v_b(a)=-7$, then $a,b\in N$ cannot be in the same coalition in a Pareto optimal partition. 
For the sake of contradiction, assume that $\pi$ is not Pareto optimal and there exists a partition $\pi'$ which Pareto dominates $\pi$. 
We will see that if there exists a player $i\in N$ such that $u_{\pi'}>u_{\pi}$, then there exists at least one $j\in N$ such that $u_{\pi'}<u_{\pi}$. 
The only players whose utility can increase (without causing some other player to be less happy) are $\{x^s\mid s\in S\}$, $\{w^s \mid s\in S\}$ or $\{z^r\mid r\in R\}$. 
We consider these player classes separately.
If the utility of player $x^s$ increases, it can only increase from $3$ to $6$ so that $x^s$ is in the same coalition as $y^s$ and $w^s$. However, this means that $y^s$ gets a decreased utility. The utility of $y^s$ can increase or stay the same only if it forms a coalition with some $z^r$s. However in that case, to satisfy all $z^r$s, there needs to exist an $S'\subseteq S$ such that $S'$ is a partition of $R$.

Assume the utility of a player $w^s$ for $s\in S$ increases. This is only possible if $w^s$ is in the same coalition as $x^s$. Clearly, the coalition formed is $\{w^s, x^s\}$ because coalition $\{w^s, x^s, y^s\}$ brings a utility of 2 to $y^s$. In that case $y^s$ needs to form a coalition $\{y^s, z_i, z_j,z_k\}$ where $s=\{i,j,k\}$. If $y^s$ forms a coalition $\{y^s, z_i, z_j,z_k\}$, then all players $y^{s'}$ for $s'\in (S\setminus \{s\})$ need to form coalitions of the form $\{y^{s'}, z_{i'}, z_{j'},z_{k'}\}$ such that $s'=\{i',j',k'\}$. Otherwise, their utility of $3$ decreases. This is only possible if there exists a set $S'\subseteq S$ of $R$ such that $S'$ is a partition of~$R$. 

Assume that there exists a partition $\pi'$ that Pareto dominates $\pi$ and the utility of a player $u_{\pi'}(z^r)>u_{\pi}(z^r)$ for some $r\in R$. This is only possible if each $z^r$ forms the coalition of the form $\{z^r,z^{r'}, z^{r''}, y^s\}$ where $s=\{r,r',r''\}$. This can only happen if there exists a set $S'\subseteq S$ of $R$ such that $S'$ is a partition of~$R$.\end{proof}

The fact that checking whether a partition is Pareto optimal is coNP-complete has no obvious implications on the complexity of 
\emph{computing} a Pareto optimal partition. In fact there is a simple polynomial-time algorithm to compute a partition which is Pareto optimal for strict preferences.

\begin{theorem}
	For strict preferences, a Pareto optimal partition can be computed in polynomial time.	
	\label{prop:PO-easyINL}
\end{theorem}

\begin{proof}
	The statement follows from an application of \emph{serial dictatorship}. \emph{Serial dictatorship}~\citep{AbSo98a} is a well-known mechanism in resource allocation in which an arbitrary player is chosen as the `dictator' who is then given his most favored allocation and the process is repeated until all players or resources have been dealt with. In the context of coalition formation, serial dictatorship is well-defined if preferences of players over coalitions are strict. 
	Serial dictatorship is also well-defined for \ASHGS with strict preferences as the dictator forms a coalition with all the players he strictly likes who have been not considered as dictators or are not already in some dictator's coalition. The resulting partition $\pi$ is such that for any other partition $\pi'$, at least one dictator will strictly prefer $\pi$ to $\pi'$. Therefore $\pi$ is Pareto optimal.
\end{proof}

A standard criticism of Pareto optimality is that it can lead to inherently unfair allocations. To address this criticism, the algorithm can be modified to obtain less lopsided partitions. Whenever an arbitrary player is selected to become the dictator among the remaining players, choose a player that does not get extremely high elitist social welfare among the remaining players. 
Nevertheless, even this modified algorithm may output an partition that fails to be individually rational.

We know that the set of partitions which are both Pareto optimal and individually rational is non-empty. Repeated Pareto improvements on individually rational partition consisting of singletons leads to a Pareto optimal and individually rational partition. We show that computing a Pareto optimal and individually rational partition for \ASHGS is weakly NP-hard.

\begin{theorem}
	Computing a Pareto optimal and individually rational partition is weakly NP-hard.
\end{theorem}
\begin{proof}
	
	Consider the decision problem {\sc SubsetSumZero} in which an instance consists of a set of $k$ integer weights $A=\{a_1, \ldots, a_k \}$ 
	\eat{such that $\sum_{a_i\in A}a_i=W$} 
	and the question is whether there exists a non-empty $S\subseteq A$ such that $\sum_{s\in S}s=0$?
	Since {\sc SubsetSum} for positive integers is NP-complete, it follows that {\sc SubsetSumZero} is also NP-complete.\footnote{We note that in any instance of {\sc SubsetSum} all zeros in the set $A$ can be omitted to obtain an equivalent problem. Reduce {\sc SubsetSum} to {\sc SubsetSumZero} by adding $a_{k+1}=-W$ to $A$.} Therefore, {\sc MaximalSubsetSumZero}, the problem of finding a maximal cardinality subset $S\subseteq A$ such that $\sum_{s\in S}s=0$ is NP-hard. 
	
	We prove the theorem by a reduction from {\sc MaximalSubsetSumZero}. Reduce an instance of $I$ of {\sc MaximalSubsetSumZero} to an instance $I'=(N,\pref)$ where $(N,\pref)$ is an \ASHG defined in the following way:
	$N=\{x,y_1,y_2\}\cup Z$ where $Z=\{z_i\mid i\in \{1,\ldots, k\}$; $v_x(y_1)=v_x(y_2)=k+1$; $v_x(z_i)=1$ for all $i\in \{1,\ldots,k \}$; $v_{y_1}(z_i)=-v_{z_i}(y_1)=-v_{y_2}(z_i)=v_{z_i}(y_2)=a_i$ for all $i\in \{1,\ldots,k\}$; and $v_{a}(b)=0$ for any $a,b\in N$ for which $v_{a}(b)$ is not already defined.
	
	First, we show that in an individually rational partition $\pi$, no player except $x$ gets positive utility, i.e., $u_\pi(b)=0$ for all $b\in N\setminus\{x\}$.
	 Assume that w.l.o.g $y_1$ gets positive utility in $\pi$. This implies there exist a subset $Z'=Z\cap \pi(y_1)$ such that $\sum_{z\in Z'}v_{y_1}(z)>0$. Then there exists $z\in Z'$ such that $v_{y_1}(z)>0$ which means that $v_{z}(y_1)<0$. Due to individual rationality, $y_2\in\pi(z)=\pi(y_1)$. But if $y_1\in\pi(y_2)$, then $u_\pi(y_2)=\sum_{z\in Z'}-v_{y_1}(z)<0$ and $\pi$ is not individually rational.
	
	Assume that there exists a $z_i\in Z$ such that $u_\pi(z_i)>0$. Then without loss of generality $v_{z_i}(y_1)>0$ and due to individual rationality $y_1\in\pi(z_i)$. Again due to individual rationality, $y_1$ needs to be with another $z_j$ such that $v_{y_1}(z_j)>0$. And again due to individual rationality, $z_j$ needs to be with $y_2$. This means, that for each $z_l\in\pi(z_i)\cap Z$, $u_\pi(z_l)=a_l-a_l=0$.

		
	We show that in every Pareto optimal and individually rational partition $\pi$, we have $y_1,y_2\in\pi(x)$.
	For any other partition $\pi'$, in which this does not hold, $u_{\pi'}(x)\leq 2k+1 < 2k+2=u_\pi(x)$.

	Consider an $S\subseteq A$ and let $\pi_z^S$ be any partition of $\{z_i\mid a_i \in A\setminus S\}$.
	\eat{
	The claim is that $\pi=\{\{x,y_1,y_2\}\cup \{z_i\mid a_i\in S\}\}\cup \pi_z^S$ is a Pareto optimal and individually rational partition if and only if $S\subseteq A$ is the maximal subset such that $\sum_{s\in S}s=0$.
	}
	The claim is that $\pi$ is a Pareto optimal and individually rational partition if and only if $\pi$  is of the form $\{\{x,y_1,y_2\}\cup \{z_i\mid a_i\in S\}\}\cup \pi_z^S$ where $S\subseteq A$ is the maximal subset such that $\sum_{s\in S}s=0$.	
\eat{	Assume that $S\subseteq A$ is the maximal subset such that $\sum_{s\in S}s=0$. We show that such $\pi$ is Pareto optimal and individually rational. Assume that there exists a partition $\pi'$ that Pareto dominates $\pi$. In order to Pareto dominate $\pi$ there must be at least one $b\in N$ for which $u_{\pi'}(b)>u_\pi(b)\geq 0$. This also implies that $\pi'$ has to be individually rational. Therefore $b=z_i$ for some $i\in\{1,\ldots,k\}$ is not possible as $u_\pi(z_i)=0$ for all $i\in\{1,\ldots,k\}$ in every individual rational partition. If $b=x$, by maximality of $S$, we know that $u_{\pi'}(y_1)\neq 0$. Then there must be a $y\in\{y_1,y_2\}$ such that $u_{\pi'}(y)<0=u_\pi(y)$. If without loss of generality $b=y_1$, we have again $u_{\pi'}(y_2)=-u_{\pi'}(y_1)<0$. 
}	
	Assume that $S\subseteq A$ is not a maximal subset such that $\sum_{s\in S}s=0$. If $\sum_{s\in S}s\neq 0$, there exists a $y\in\{y_1,y_2\}$ such that $u_\pi(y)<0$. If $S$ is not maximal then there is a larger set $S'$ and a corresponding partition $\pi'=\{\{x,y_1,y_2\}\cup \{z_i\mid a_i\in S'\}\}\cup \pi_z^{S'}$ with $u_\pi(x)=|S|<|S'|=u_{\pi'}(x)$ and $u_\pi(b)=u_{\pi'}(b)$ for all $b\in N\setminus \{x\}$.
	For any other $S'\subseteq A$ such that $|S'|>|S|$, we know that $\sum_{s'\in S'}\leq 0$ which implies that there is a $y\in\{y_1,y_2\}$ which gets negative utility.
\end{proof}

\section{Complexity of envy-freeness}~\label{sec:envy}

Envy-freeness is a desirable property in resource allocation, especially in \emph{cake cutting} settings. \citet{LMMS04a} proposed envy-minimization in different ways and examined the complexity of minimizing envy in resource allocation settings. \citet{BoJa02a} mentioned envy-freeness in hedonic games but focused on stability. We already know that envy-freeness can be easily achieved by the partition of singletons.\footnote{The partition of singletons also satisfies individual rationality.} Therefore, in conjunction with envy-freeness, we seek to satisfy other properties such as stability or Pareto optimality. A partition is \emph{Nash stable} if there is no incentive for a player to be deviate to another (possibly empty) coalition.
For symmetric \ASHGS, it is known that Nash stable partitions always exist and they correspond to partitions for which the utilitarian social welfare is a local optimum~\citep[see, \eg][]{BoJa02a}. We now show that for symmetric \ASHGS, there may not exist any partition which is both envy-free and Nash stable. 

\begin{example}\label{example:ns-envy}
Consider an \ASHG $(N,\pref)$ where $N=\{1,2,3\}$ and $\pref$ is defined as follows: $v_1(2)=v_2(1)=3$, $v_1(3)=v_3(1)=3$ and $v_2(3)=v_3(2)=-7$. Then there exists no partition which is both envy-free and Nash stable. 
\end{example}

We use the game in Example~\ref{example:ns-envy} as a gadget to prove the following.\footnote{Example~\ref{example:ns-envy} and the proof of Theorem~\ref{th:envy-nash} also apply to the combination of envy-freeness and individual stability where individual stability is a variant of Nash stability~\citep{BoJa02a}.} 

\begin{theorem}\label{th:envy-nash}
For symmetric preferences, checking whether there exists a partition which is both envy-free and Nash stable is NP-complete in the strong sense. 
\end{theorem}
\begin{proof}

The problem is clearly in NP since envy-freeness and Nash stability can be verified in polynomial time. We reduce the problem from E3C. Let $(R,S)$ be an instance of E3C where $R$ is a set and $S$ is a collection of subsets of $R$ such that $|R|=3m$ for some positive integer $m$ and $|s| = 3$ for each $s \in S$. We will use the fact that E3C remains NP-complete even if each $r\in R$ occurs in at most three members of $S$. 
$(R,S)$ can be reduced to an instance $(N,\pref)$ where $(N,\pref)$ is an \ASHG defined in the following way. Let $N=\{y^s \mid  s \in S\} \cup \{z_{1}^r, z_{2}^r, z_{3}^r \mid  r\in R \}$. We set all preferences as symmetric. The players preferences are as follows: for all $r\in R$, $v_{z_1^r}(z_2^r)=v_{z_2^r}(z_1^r)=3$, $v_{z_1^r}(z_3^r)=3$ and $v_{z_2^r}(z_3^r)=v_{z_3^r}(z_2^r)=-7$; for all $s=\{i,j,k\}\in S$, $v_{z_1^i}(z_1^j)=v_{z_1^i}(z_1^k)=v_{z_1^j}(z_1^k)=1/10$ and $v_{y^s}(z_1^i)=v_{y^s}(z_1^j)=v_{y^s}(z_1^k)=28/10$; 
and for all $a,b\in N$ for which valuations have not been defined, $v_a(b)=v_b(a)=-7$

\eat{
	\begin{itemize}
	\item For all $r\in R$, $v_{z_1^r}(z_2^r)=v_{z_2^r}(z_1^r)=3$, $v_{z_1^r}(z_3^r)=3$ and $v_{z_2^r}(z_3^r)=v_{z_3^r}(z_2^r)=-7$.
	\item For all $s=\{i,j,k\}\in S$, $v_{z_1^i}(z_1^j)=v_{z_1^i}(z_1^k)=v_{z_1^j}(z_1^k)=v_{y^s}(z_1^i)=v_{y^s}(z_1^j)=v_{y^s}(z_1^k)=1$.
	\item For all $a,b\in N$ for which valuations have not been defined, $v_a(b)=v_b(a)=-7$
    \end{itemize}
}
We note that $-7$ is a sufficiently large negative valuation to ensure that if $v_a(b)=v_b(a)=-7$, then $a$ and $b$ will get negative utility if they are in the same coalition. We show that there exists an envy-free and Nash stable partition for $(N,\pref)$ if and only if $(R,S)$ is a `yes' instance of E3C.

Assume that there exists $S'\subseteq S$ such that $S'$ is a partition of $R$. Then there exists a partition $\pi=$
$\{\{y^s, z_1^i, z_1^j, z_1^k\}\mid s=\{i,j,k\}\in S' \} \cup \{\{z_2^r\}, \{z_3^r\}\mid r\in R\}\cup \{\{s\}\mid s\in S\setminus S'\}$. It is easy to see that partition $\pi$ is Nash stable and envy-free. Players $z_1^r$ and $z_3^r$ both had an incentive to be with each other when they are singletons. However, each $z_1^r$ now gets utility 3 by being in a coalition with $z_1^{r'}$, $z_1^{r''}$ and $y^s$ where $s=\{r,r',r''\}\in S$. Therefore $z_1^r$ has no incentive to be with $z_3^r$ and $z_3^r$ has no incentive to join $\{z_1^{r'},z_1^{r'},z_1^{r''}, y^s\}$ because $v_{z_3^r}(z_1^{r'})=v_{z_3^r}(z_1^{r''})=v_{z_3^r}(y^s)=-7$. Similarly, no player is envious of another player.

Assume that there exists no partition $S'\subseteq S$ of $R$ such that $S'$ is a partition of $R$. Then, there exists at least one $r\in R$ such that $z_1^i$ is not in the coalition of the form $\{z_1^{r},z_1^{r'},z_1^{r''}, y^s\}$ where $s=\{r,r',r''\}\in S$. Then the only individually rational coalitions which $z_1^r$ can form and get utility at least $3$ are the following $\{z_1^r, z_3^r\}$, $\{z_1^r, z_2^r\}$. In the first case, $z_1^r$ wants to deviate to $\{z_3^r\}$. In the second case, $z_2^r$ is envious and wants to replace  $z_3^{r}$. Therefore, there exists no partition which is both Nash stable and envy-free.
\end{proof}

While the existence of a Pareto optimal partition and an envy-free partition is guaranteed, we show that checking whether there exists a partition which is both envy-free and Pareto optimal is hard.

\begin{theorem}
	Checking whether there exists a partition which is both Pareto optimal and envy-free is $\Sigma_{2}^{p}$-complete.
	\label{thm:envy-po-hard}
\end{theorem}

\begin{proof}
    The problem has a `yes' instance if there exists an envy-free partition that Pareto dominates every other partition. Therefore the problem is in the complexity class ${\mathrm{NP}}^{\mathrm{coNP}}=\Sigma_{2}^{p}$.

    We prove hardness by a reduction from a problem concerning resource allocation (with additive utilities) \citep{KBKZ09a}. A resource allocation problem is a tuple $(I,X,w)$ where $I$ is a set of agents, $X$ is a set of indivisible objects and $w:I\times X \rightarrow \mathbb{R}$ is a weight function. An $a:I\rightarrow 2^X$ is an allocation if for all $i,j\in I$ such that $i\neq j$, we have $a(i)\cap a(j)=\emptyset$. The resultant utility of each agent $i\in I$ is then $\sum_{x\in a(i)} w(i,x)$.
It was shown by \citet{KBKZ09a} that the problem $\exists$-EEF-ADD of checking the existence of an envy-free and Pareto optimal allocation is $\Sigma_{2}^{p}$-complete.

Now, consider an instance $(I,X,w)$ of $\exists$-EEF-ADD and reduce it to an instance $(N,\pref)$ of an \ASHG where $N=I\cup X$ and $\pref$ is specified by the following values: $v_i(x_j)=w(i,x_j)$ and $v_{x_j}(i)=0$ for all $i\in I$, $x_j \in X$; $v_{x_k}(x_j)=v_{x_j}(x_k)=0$ for all $x_j,x_k$; and $v_i(j)=v_j(i)=-W\cdot|I\cup X|$  for all $i,j\in I$ where $W=\sum_{i\in I,x_j\in X}|w(i,x_j)|$.
\eat{
	\begin{itemize}
	\item For all $i\in I$, $x_j \in X$, $v_i(x_j)=w(i,x_j)$ and $v_{x_j}(i)=0$.
	\item For all $x_j,x_k$, $v_{x_i}(x_j)=v_{x_j}(x_i)=0$.
	\item For all $i,j\in I$, $v_i(j)=v_j(i)=-W|I\cup X|$.
	\end{itemize}
	}
	It can then be shown that there exists a Pareto optimal and envy-free partition in $(N,\pref)$ if and only if $(I,X,w)$ is a `yes' instance of $\exists$-EEF-ADD. The proof is omitted due to space limitations.\end{proof}
	
\eat{It is clear that for any Pareto optimal partition $\pi$, there exist no $i,j\in I \subset N$ such that $i\neq j$ and $j\in \pi(i)$. Assume that this were not the case and there exist $i,j\in I \subset N$ such that $i\neq j$ and $j\in \pi(i)$. Then $i$ and $j$ both get negative value because $\sum_{k\in \pi(i)}v_i(k)=\sum_{k\in (\pi(i)\setminus \{j\})}v_i(k)-W<0$ and 
	$\sum_{k\in \pi(i)}v_j(k)=\sum_{k\in (\pi(i)\setminus \{i\})}v_j(k)-W<0$. 
	Then $i$ and $j$ can be separated to form singletons to get another partition $\pi'$, where the value of every other player $k\in (N\setminus \{i,j\})$ gets the same value while $i$ and $j$ get at least zero value. Therefore there is a one-to-one correspondence between any such partition $\pi$ and allocation $a$ where $a(i)=\pi(i)\setminus \{i\}$. It now easy to see that $\pi$ is Pareto optimal and envy-free in $G$ if and only if $a$ is a Pareto optimal and envy-free allocation.}

The results of this section show that, even though envy-freeness can be trivially satisfied on its own, it becomes much more delicate when considered in conjunction with other desirable properties.

\section{Conclusions}

We studied the complexity of partitions that satisfy standard criteria of fairness and optimality in additively separable hedonic games. 
We showed that computing a partition with maximum egalitarian or utilitarian social welfare is NP-hard in the strong sense and computing an individually rational and Pareto optimal partition is weakly NP-hard. A Pareto optimal partition can be computed in polynomial time when preferences are strict. Interestingly, checking whether a given partition is Pareto optimal is coNP-complete even in the restricted setting of strict and symmetric preferences.

We also showed that checking the existence of partition which satisfies not only envy-freeness but an additional property like Nash stability or Pareto optimality is computationally hard. The complexity of computing a Pareto optimal partition for \ASHGS with general preferences is still open. 
Other directions for future research include approximation algorithms to compute maximum utilitarian or egalitarian social welfare for different representations of hedonic games.




\begin{contact}
Haris Aziz, Felix Brandt, and Hans Georg Seedig\\
Department of Informatics\\ 
Technische Universit\"at M\"unchen\\
85748 Garching bei M\"unchen, Germany\\
\texttt{\small\{aziz,brandtf,seedigh\}@in.tum.de}
\end{contact}

\end{document}